\documentclass{article}

\usepackage{PRIMEarxiv}

\usepackage[utf8]{inputenc} 
\usepackage[T1]{fontenc}    
\usepackage{hyperref}       
\usepackage{url}            
\usepackage{booktabs}       
\usepackage{amsfonts}       
\usepackage{amssymb}
\usepackage{amsmath}
\usepackage{nicefrac}       
\usepackage{microtype}      
\usepackage{fancyhdr}       
\usepackage{graphicx}       
\usepackage{algorithm}
\usepackage{algpseudocode}
\usepackage{amsthm}
\newtheorem{definition}{Definition}
\newtheorem{theorem}{Theorem}
\usepackage{tikz}
\usepackage{array}
\newcolumntype{L}[1]{>{\raggedright\arraybackslash}p{#1}}
\usetikzlibrary{shapes,arrows,positioning,fit, arrows.meta, calc}
\graphicspath{{media/}}     

\title{Weird Machines in Transport Layer Security}

\author{
Michael Collins \\
Laboratory for Advanced Cybersecurity Research \\
Fort Meade, MD 20755 \\
\texttt{mdcolli@uwe.nsa.gov} \\
\And
Jada Cumberland \\
School of Cybersecurity \\
Old Dominion University \\
Norfolk, VA 23259 \\
\texttt{jcumb002@odu.edu} \\
\And
Brianne Dunn \\
Old Dominion University \\
Norfolk, VA 23259 \\
\texttt{bdunn003@odu.edu} \\
\And
Ross Gore \\
Center for Secure and Intelligent Critical Systems \\
Old Dominion University \\
Norfolk, VA 23259 \\
\texttt{rgore@odu.edu} \\
\And    
Samuel Jackson \\
Old Dominion University \\
Norfolk, VA 23259 \\
\texttt{sjack047@odu.edu} \\
\And
Sachin Shetty \\
Center for Secure and Intelligent Critical Systems \\
Department of Electrical and Computer Engineering \\
Old Dominion University \\
Norfolk, VA 23259 \\
\texttt{sshetty@odu.edu} \\
\And
Jonathan Takeshita \\
Department of Computer Science \\
School of Cybersecurity \\
Old Dominion University \\
Norfolk, VA 23259 \\
\texttt{jtakeshi@odu.edu} \\
}

\begin{document}
\maketitle

\begin{abstract}
Weird machines are latent computational capabilities that emerge from the composition of architectural components. Prior work has studied this phenomenon extensively in software systems, including x86 instructions, ELF metadata, and page tables, and more recently in cyber-physical systems such as industrial control networks. This paper extends weird machine theory to a new domain: the Transport Layer Security (TLS) handshake and its two dominant implementations, OpenSSL and BoringSSL.

We show that legitimate TLS primitives, including session cache entries, renegotiation logic, extension parsing, and certificate verification steps, compose into Turing-complete systems whose computation is coupled to authentication and trust decisions rather than physical actuation. We formalize this coupling, which we call trust actuation, and argue that any TLS implementation providing session storage, arithmetic on sequence counters, conditional branching on handshake state, and iteration through resumption or retry loops satisfies the conditions for arbitrary computation.

We validate this theory with two working demonstrations built on real OpenSSL code paths. The first, a sentinel system, composes standard TLS primitives into a defensive mechanism that detects anomalous handshake behavior. The second, an authentication bypass, composes the same class of primitives into an attack that defeats a cipher-strength policy check through mid-connection renegotiation, without any memory corruption or external malware. Both demonstrations run against real server and client binaries in Docker.

We then examine OpenSSL and BoringSSL, arguing from each library's documented design that BoringSSL removes or restricts several gadgets, such as default renegotiation, that OpenSSL leaves enabled for compatibility. This divergence suggests that library design choices, not just implementation bugs, shape the available attack and defense surface, though we validate this claim only against OpenSSL and treat the BoringSSL comparison as a documented hypothesis rather than an empirical result. Our results indicate that dual-use weird machine capability is not unique to cyber-physical systems. It is a general property of any protocol implementation with sufficient state, and TLS libraries must be evaluated with this in mind.
\end{abstract}

\keywords{Weird machines \and Transport Layer Security \and Turing completeness \and Trust actuation \and TLS renegotiation \and Authentication bypass}

\section{Introduction}

The concept of weird machines, referring to unintended computational systems arising from the composition of program components, emerged from the software exploitation community \cite{bratus2011exploit}. Researchers showed that software artifacts such as x86 MOV instructions \cite{domas2013mov}, ELF file metadata \cite{shapiro2013weird}, and page fault handling \cite{bangert2013printable} can be composed to perform arbitrary computation independent of a program's original intent. This prior work examined weird machines in purely computational contexts. The consequences of exploitation were limited to information disclosure or unauthorized computation.

Industrial control systems extended this idea to a domain where computation controls physical processes. Architectural primitives such as PLC timers and Modbus registers \cite{modbus2012protocol} compose into Turing-complete systems that can both cause harm and provide resilience. This paper extends the same question to a different domain: network security protocols, and specifically the Transport Layer Security (TLS) handshake as implemented in OpenSSL and BoringSSL.

TLS implementations present a different substrate than either software exploits or industrial control systems. State in a TLS session includes handshake progress, session cache entries, and certificate verification results. These are not physical quantities, but they carry consequences once a decision is made. Once a server authenticates a client or resumes a session, that decision persists and grants access. We call this property trust actuation. Computation in a TLS library does not move a motor or trip a relay. It grants or denies trust.

The recognition that security mechanisms can exhibit unintended computational power is not new. Harrison, Ruzzo, and Ullman (HRU) showed in 1976 that the safety question for access control matrices is undecidable \cite{harrison1976protection}. They proved this by encoding a Turing machine as a protection system. Subjects represented tape cells. Access rights represented tape symbols and machine states. State transitions were implemented through protection system commands \cite{bishop2002computer}. If an analyzer could determine whether a specific right would ever appear in a matrix cell, it would solve the Halting Problem. This result established that even carefully designed security systems can harbor unexpected computational power.

The HRU construction was deliberately artificial. The conditions needed to build a Turing machine from an access control matrix were impractical and never intended for deployment. Our work differs from HRU in one important way. TLS state machines are not artificial constructions. They exist in every production implementation of the protocol. Session caches, renegotiation logic, and certificate verification chains are required features, not contrived encodings. Any TLS library that supports session resumption and certificate-based authentication already provides the primitives needed for arbitrary computation.

Lamport's work on the glitch phenomenon and Buridan's Principle offers a useful frame for this observation \cite{lamport1974glitch, lamport1984buridan}. Lamport examined the RS flip-flop and observed that the metastable state occurs by design, not by defect. The state is declared undefined in logical specifications, yet it is an inherent property of the circuit's physical implementation. Applied to TLS, this principle suggests that weird machine potential is not a bug to be patched. It is a consequence of providing memory, arithmetic, conditionals, and iteration in a protocol state machine. A TLS library that removes these primitives loses functionality it needs for correct operation. A TLS library that keeps them retains the conditions for weird machine construction.

History already contains a documented case of a flawed TLS state machine. Beurdouche et al. tested popular open source TLS implementations and found that composite state machines allowed message sequences the protocol never intended \cite{beurdouche2015smack}. Their SMACK and related findings showed that state machine bugs enabled authentication bypasses across multiple libraries. CVE-2020-2655 showed a similar pattern in a Java TLS implementation, which could be tricked into skipping client authentication entirely by manipulating handshake message order \cite{cve2020_2655}. These are not new vulnerability classes. They are existing evidence that TLS state composition, not memory corruption, can defeat authentication.

Pavlovic and Seidel's work on security science provides a second useful lens \cite{pavlovic2025security}. They note a basic asymmetry in security claims. A precise security claim can be tested and proven false, but it can never be proven true. This asymmetry applies directly to TLS weird machines. When we construct a working authentication bypass from legitimate TLS primitives, we prove a specific case of insecurity. We do not need a bug in the traditional sense. We only need to show that the architecture allows a computation the designers did not intend.

OpenSSL and BoringSSL provide a natural pair for this study. Both implement the same protocol, but they diverge in design philosophy. OpenSSL enables renegotiation by default and accepts renegotiation requests from a peer without extra configuration \cite{boringssl2024porting}. BoringSSL removes this behavior. Google's engineers consider renegotiation an extremely problematic feature, and BoringSSL's smaller, more opinionated codebase reflects that judgment \cite{fastly2024boringssl}. This divergence means the same weird machine gadget can exist in one library and not the other, even though both libraries implement the same protocol standard.

This paper makes four contributions. First, we extend weird machine theory to TLS protocol implementations, formalizing trust actuation as the analog to physical actuation. Second, we define an eight-category gadget taxonomy specific to TLS implementations, covering read/write access, control-flow, communication bridging, security processing, timing and synchronization, arithmetic and computation, gadget composition and chaining, and externally visible I/O and side effects. Third, we build and validate two working demonstrations against real OpenSSL code paths: a defensive sentinel system and an authentication bypass, showing that the same class of gadgets supports both beneficial and adversarial outcomes. Fourth, we examine documented design differences between OpenSSL and BoringSSL and argue that library design choices, not implementation bugs alone, shape which weird machines are constructible, though we validate this claim only against OpenSSL.

The remainder of this paper proceeds as follows. Section~\ref{sec:background} reviews prior work on weird machines, TLS state machine vulnerabilities, and the history of OpenSSL and BoringSSL. Section~\ref{sec:theory} develops the theoretical framework, including the TLS gadget taxonomy and the argument that sufficient TLS primitives yield Turing completeness. Section~\ref{sec:testbed} describes our testbed and the gadgets available in it. Sections~\ref{sec:sentinel} and \ref{sec:authbypass} present the sentinel and authentication bypass demonstrations. Section~\ref{sec:dualuse} analyzes the dual-use property of the gadgets we identify, Section~\ref{sec:limitations} states the boundaries of these results, and Section~\ref{sec:conclusion} concludes with implications for TLS library design.

\section{Background and Related Work}
\label{sec:background}

\subsection{Turing Completeness in Security Systems}

The theoretical foundations for understanding computational emergence in security systems trace to Harrison, Ruzzo, and Ullman's (HRU's) 1976 proof that the safety problem for protection systems is undecidable \cite{harrison1976protection}. They demonstrated this result by encoding a Turing machine as an access control matrix. Subjects represented tape cells. Rights corresponded to tape symbols and machine states. Commands implemented state transitions. Any protection system capable of determining whether a specific right would ever appear in a matrix cell would thereby solve the Halting Problem.

The HRU result established that security analysis has fundamental limits. As Bishop notes in his computer security text \cite{bishop2002computer}, the construction required carefully crafted conditions and was deliberately artificial. The assumptions needed to build a Turing machine from an access control matrix were impractical for real systems. The result was primarily theoretical. It proved that perfect safety analysis is impossible, but the construction was never intended for deployment.

Our contribution differs from HRU because we examine Turing completeness in operational systems designed for deployment. TLS implementations provide memory, arithmetic, conditionals, and iteration not through contrived theoretical conditions but through the essential requirements of secure communication. The primitives enabling Turing completeness in TLS are not artifacts of clever encoding. They are necessary components for protocol functionality.

\subsection{Weird Machines in Software and Hardware Systems}

The weird machine concept emerged from exploit development research when Bratus et al. formalized the idea that buffer overflows and other memory corruption create unintended state machines programmable through attacker-controlled input \cite{bratus2011exploit}. Subsequent work demonstrated weird machine construction across several software contexts. At the instruction level, Domas showed that x86 MOV instructions alone are Turing-complete \cite{domas2013mov}. This result challenged assumptions about instruction set expressiveness and its relevance to exploitation. At the metadata level, Shapiro et al. demonstrated computation embedded in ELF file headers \cite{shapiro2013weird}. By crafting metadata fields such as section tables and relocation entries, an attacker can implement logic gates and control flow without executing traditional code. Bangert et al. extended this line of work to page fault handling, showing that page table manipulation alone supports instruction-less computation \cite{bangert2013printable}. Dullien later formalized the relationship between weird machines, exploitability, and provable unexploitability, arguing that some classes of weird machines can be shown not to exist under specific constraints \cite{dullien2018weird}.

These software results share a common limitation. They examine systems that operate entirely within the computational domain. Exploitation affects memory, control flow, or information leakage, but never a decision with external consequence such as authentication.

Fault injection research extends this idea into physical hardware exploitation. Voltage glitching, clock glitching, and related techniques induce faults that bypass security checks such as secure boot \cite{giller2015glitching}. The Xbox 360 reset glitch hack demonstrated practical fault injection using equipment costing under one hundred dollars \cite{xbox360glitch}. Murdock et al. showed that software-controlled voltage faults can extract secrets from trusted execution environments such as Intel SGX \cite{murdock2021weird}. Living Off the Land attacks represent a related pattern at a higher level of abstraction. These attacks use legitimate, preinstalled system tools for malicious purposes without introducing external malware \cite{lolbins2025crowdstrike}. The TA505 threat group used this technique extensively in 2018 phishing campaigns against financial institutions \cite{ta505lotl}. Living Off the Land attacks and weird machines share the same underlying principle. Both use intended functionality in unintended composition, and both create a dual-use problem in which removing the capable tool also removes legitimate functionality.

\subsection{Metastability and Beneficial Emergence}

Lamport's original paper with Palais established what they called the Principle of the Glitch. For any device making a discrete decision based on continuous inputs, there exist inputs that cause arbitrarily long decision times \cite{lamport1974glitch}. Lamport later formalized this observation as Buridan's Principle \cite{lamport1984buridan}. Anderson and Gouda extended this work to purely digital systems, proving that arbitrarily long metastable states can occur even in discrete domains \cite{anderson1991glitch}. Lamport's later work on interprocess communication further developed how systems handle these unavoidable timing ambiguities in distributed settings \cite{lamport1985interprocess}.

In examining the RS flip-flop, Lamport observed that when Set and Reset inputs transition simultaneously, the circuit enters an undefined metastable state where outputs oscillate unpredictably. This behavior is not a bug. It is an architectural property. The metastable state exists because of how flip-flops are designed, not despite it. This principle applies directly to any system that must make a discrete decision, including a TLS server deciding whether to accept a handshake message in a given state.

History provides a documented case of beneficial emergent computation under crisis conditions. Following the 1970 Apollo 13 oxygen tank explosion, the crew needed a way to remove carbon dioxide using components never designed to fit together \cite{lovell1994lost, edge2020apollo}. NASA's ground team built an adapter from suit hoses, plastic bags, and duct tape, then verbally transmitted assembly instructions to the crew \cite{nasa2008mailbox}. This improvised system saved three lives. It shows that architectural flexibility, when paired with well-understood components, can produce beneficial results that were never part of the original specification. This precedent motivates our claim that the same gadgets that enable a TLS authentication bypass can also enable a legitimate defensive mechanism.

\subsection{TLS Protocol and Implementation Landscape}

TLS secures the majority of traffic on the modern internet. The current version of the protocol, TLS 1.3, defines a handshake state machine with specific rules about message order, key derivation, and session resumption \cite{rfc8446}. Two implementations dominate production use today. OpenSSL is the general-purpose library used across most Linux servers and a wide range of applications. BoringSSL is Google's internal fork, created in 2014 after the Heartbleed vulnerability exposed a severe flaw in OpenSSL's heartbeat extension handling \cite{heartbleed2014}.

The two libraries diverge sharply in design philosophy. OpenSSL prioritizes backward compatibility and broad protocol support, including legacy features that many deployments no longer need \cite{stackademic2024differences}. BoringSSL prioritizes a smaller, more opinionated codebase suited to Google's internal needs, and it removes features that OpenSSL keeps for compatibility \cite{cossacklabs2017replacing}. One clear example is TLS renegotiation. OpenSSL enables renegotiation by default and accepts renegotiation requests from a peer without additional configuration \cite{boringssl2024porting}. BoringSSL removes this behavior entirely, and Google's engineers describe renegotiation as an extremely problematic feature \cite{fastly2024boringssl}. This single design decision changes which state transitions are reachable in each library, which matters directly for the gadget analysis in Section~\ref{sec:theory}.

\subsection{Prior TLS State Machine Vulnerability Research}

Prior research has already shown that flawed TLS state machines lead to authentication bypass. Beurdouche et al. systematically tested popular open source TLS implementations and discovered that composite state machines allowed message sequences the protocol specification never intended \cite{beurdouche2015smack}. Their SMACK findings demonstrated that a server could be tricked into skipping required handshake steps, including client authentication, by sending messages in an order the implementation did not anticipate. CVE-2020-2655 documented a similar flaw in a Java TLS implementation, where manipulating handshake message order allowed complete client authentication bypass \cite{cve2020_2655}.

These results establish an important precedent for our work. Both cases involved no memory corruption and no external malware. The vulnerability existed entirely in the composition of legitimate protocol states. This is exactly the property we formalize as a weird machine in Section~\ref{sec:theory}, and our authentication bypass demonstration in Section~\ref{sec:authbypass} builds on the same class of flaw using a purpose-built vulnerable server.

\subsection{Positioning This Work}

Our contribution is best understood relative to three lines of prior work, rather than as a single checklist of features no other line of work has. HRU proved Turing completeness in a theoretical access control construction never intended for deployment \cite{harrison1976protection}. Software weird machine research demonstrated Turing completeness in operational systems, but stayed entirely within the computational domain, where exploitation affects memory or control flow without an external consequence such as an authentication decision \cite{bratus2011exploit, domas2013mov, shapiro2013weird}. Prior TLS state machine research demonstrated authentication bypass through message reordering, but did not frame the result as an instance of weird machine theory, nor connect it to a broader gadget taxonomy \cite{beurdouche2015smack, cve2020_2655}.

Our work sits at the intersection of these three lines rather than superseding any of them. Like the software weird machine literature, we examine an operational system rather than a theoretical construction. Like the TLS state machine literature, our authentication bypass in Section~\ref{sec:authbypass} produces a working authentication bypass through message and event ordering rather than memory corruption. What we add is a connection between these two lines: we organize the primitives involved into an explicit gadget taxonomy, argue that the resulting composition satisfies the conditions for Turing completeness, following the same proof-sketch style HRU itself used \cite{harrison1976protection}, and build a working defensive counterpart alongside the exploit to show the same gadgets support both outcomes. We also examine documented design differences between OpenSSL and BoringSSL and argue, based on each library's public documentation, that these differences change which of our gadgets are available in each library. We validate this argument only against OpenSSL, as Section~\ref{sec:limitations} discusses in detail, and it should be read as a documented hypothesis about BoringSSL rather than a result established with the same rigor as the OpenSSL demonstrations in Sections~\ref{sec:sentinel} and \ref{sec:authbypass}.

\section{Theoretical Framework}
\label{sec:theory}

This section extends weird machine theory from physical systems to protocol implementations, establishing formal foundations for trust-coupled computation in TLS.

\subsection{From Physical Actuation to Trust Actuation}

Prior weird machine research examined systems where computation either stayed within a computer's memory or directly moved a physical actuator such as a motor or relay. TLS occupies a third position. State in a TLS session, such as handshake progress, a cached session ticket, or a verified certificate chain, is not a physical quantity, but it is not purely internal either. It has a consequence outside the computation that produced it. Once a server accepts a Finished message or resumes a session from a stored ticket, that decision grants network access to a peer. We call this property trust actuation.

Trust actuation differs from physical actuation in one important way. A motor position can be measured directly and independently of the controller that set it. A trust decision cannot. The only record that a client was authenticated is the state the TLS library itself maintains. This makes TLS weird machines harder to detect than their cyber-physical counterparts, since the same component that performs the computation also reports on its own correctness.

\subsection{Architectural Inevitability}

Following Lamport's analysis of metastability, we treat weird machine capability in TLS as architecturally inevitable rather than as an implementation defect \cite{lamport1974glitch, lamport1984buridan}. In Lamport's analysis of the RS flip-flop, designers intend two stable states corresponding to Set and Reset. The metastable state, where both outputs oscillate unpredictably, is declared undefined in the specification but exists in every physical implementation. This state cannot be removed without changing the circuit's basic architecture.

TLS state machines exhibit the same property. Designers intend specific handshake sequences, such as ClientHello followed by ServerHello followed by Finished \cite{rfc8446}. Arbitrary computation, which permits sequences the designers did not intend, is not part of the specification, but it exists in every implementation that provides session storage, arithmetic on sequence numbers, conditional checks on handshake state, and a loop that processes incoming messages. This capability cannot be removed without limiting what the protocol can do. A library that keeps session resumption, renegotiation, or extension parsing keeps the primitives that make weird machine construction possible.

\subsection{Architectural Gadgets in TLS Implementations}
\label{sec:gadgettypes}

We categorize the primitives available in a TLS implementation into eight gadget types, summarized in Table \ref{tab:tlsgadgets8}.

\begin{table}[h]
\centering
\caption{Architectural Gadgets in TLS Implementations}
\label{tab:tlsgadgets8}
\renewcommand{\arraystretch}{1.15}
\begin{tabular}{lp{9.5cm}}
\toprule
Gadget Type & Description \\
\midrule
Read/Write & Live reads of connection state, such as the negotiated cipher, protocol version, or peer certificate, exposed entirely through the public API \\
Control-Flow & Conditional branching driven by handshake state or by a cached decision, determining how the connection proceeds \\
Communication-Bridge & The point where a raw transport socket is bound to the TLS layer, carrying bytes between the network and the protocol state machine \\
Security-Processing & Verification and policy logic, including certificate chain validation and cipher strength checks \\
Timing/Synchronization & Callback hooks tied to handshake events, such as handshake start, handshake completion, and renegotiation \\
Arithmetic/Computation & Numeric and cryptographic operations, including digest computation, byte comparison, and bit-length checks \\
Composition/Chaining & The wiring between gadgets, where one gadget's output becomes the trigger or input for another \\
I/O and Side-Effect & Externally visible actions, such as logging, terminating a connection, or releasing application data \\
\bottomrule
\end{tabular}
\end{table}

A table alone does not show when each gadget type becomes relevant during a connection's life. Figure \ref{fig:gadgettimeline} places these eight categories along the TLS handshake timeline, from the initial TCP connection through Finished, and through an optional renegotiation or resumption event.

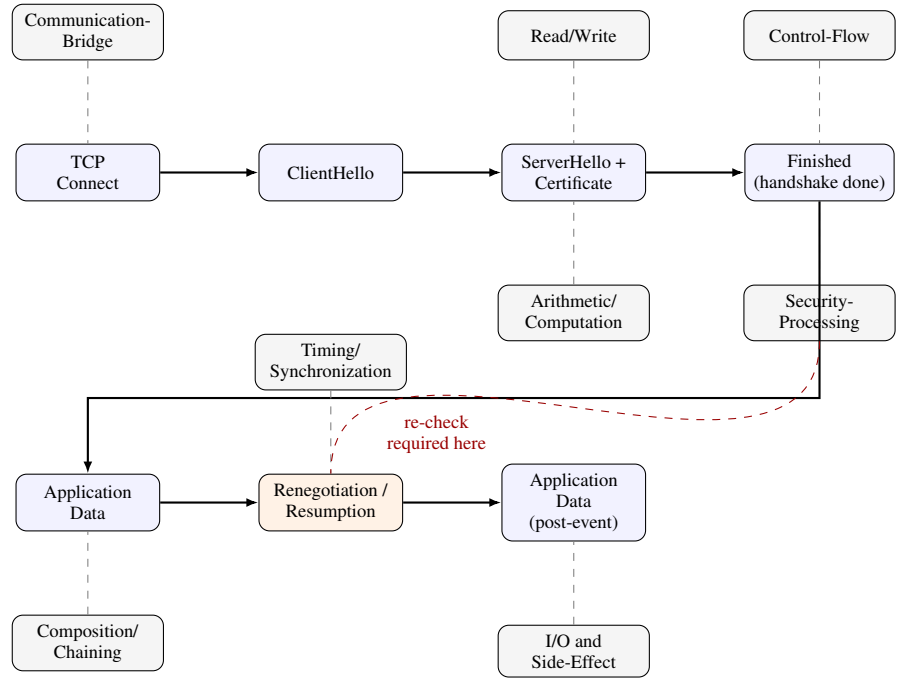
\begin{figure}[h]
\centering
\begin{tikzpicture}[
  node distance=1.3cm,
  stage/.style={draw, rectangle, rounded corners, minimum width=1.9cm, minimum height=0.75cm, align=center, font=\scriptsize, fill=blue!5},
  branchstage/.style={stage, fill=orange!10},
  gadget/.style={draw, rectangle, rounded corners, minimum width=2.0cm, minimum height=0.65cm, align=center, font=\scriptsize, fill=gray!8},
  arr/.style={-{Latex[length=1.6mm]}, thick},
  leader/.style={-{}, dashed, gray}
]

\node[stage] (tcp) {TCP\\Connect};
\node[stage, right=of tcp] (ch) {ClientHello};
\node[stage, right=of ch] (sh) {ServerHello +\\Certificate};
\node[stage, right=of sh] (fin) {Finished\\(handshake done)};

\draw[arr] (tcp) -- (ch);
\draw[arr] (ch) -- (sh);
\draw[arr] (sh) -- (fin);

\node[gadget, above=1.1cm of tcp] (g1) {Communication-\\Bridge};
\node[gadget, above=1.1cm of sh] (g2) {Read/Write};
\node[gadget, above=1.1cm of fin] (g3) {Control-Flow};

\draw[leader] (g1) -- (tcp);
\draw[leader] (g2) -- (sh);
\draw[leader] (g3) -- (fin);

\node[gadget, below=1.1cm of sh] (g5) {Arithmetic/\\Computation};
\node[gadget, below=1.1cm of fin] (g6) {Security-\\Processing};

\draw[leader] (g5) -- (sh);
\draw[leader] (g6) -- (fin);

\node[stage, below=3.6cm of tcp] (app) {Application\\Data};
\node[branchstage, right=of app] (reneg) {Renegotiation /\\Resumption};
\node[stage, right=of reneg] (app2) {Application\\Data\\(post-event)};

\draw[arr] (fin.south) -- ++(0,-2.6cm) -| (app.north);

\draw[arr] (app) -- (reneg);
\draw[arr] (reneg) -- (app2);

\node[gadget, above=1.1cm of reneg] (g4) {Timing/\\Synchronization};

\draw[leader] (g4) -- (reneg);

\node[gadget, below=1.1cm of app] (g7) {Composition/\\Chaining};
\node[gadget, below=1.1cm of app2] (g8) {I/O and\\Side-Effect};

\draw[leader] (g7) -- (app);
\draw[leader] (g8) -- (app2);

\draw[leader, red!60!black] (g6.south) to[out=-90, in=90] (reneg.north);
\node[above=0.15cm of reneg, xshift=1.4cm, font=\scriptsize, red!60!black, align=center] {re-check\\required here};

\end{tikzpicture}
\caption{Gadget Taxonomy Mapped Onto the TLS Handshake Timeline. Security-Processing and Timing/Synchronization gadgets, normally exercised once during the initial handshake, become available a second time at any renegotiation or resumption event.}
\label{fig:gadgettimeline}
\end{figure}

The red arrow in Figure \ref{fig:gadgettimeline} marks the point that matters most for the rest of this paper. Security-Processing gadgets are exercised once when an application first decides whether to trust a connection. Renegotiation and resumption reopen this same gadget a second time, since both events can change the property, such as cipher strength or certificate identity, that the original decision depended on. Nothing in the TLS specification or in OpenSSL requires an application to exercise Security-Processing again at this second point. Whether an implementation does so is a Composition/Chaining decision made entirely by the application, not by the library. Section~\ref{sec:testbed} grounds this observation in real code, showing one implementation that makes this second check and one that does not.

\subsection{Protocol-Interfaced Turing Completeness}

We now formalize the conditions under which a TLS implementation supports Turing-complete computation coupled with trust actuation.

\begin{definition}[Protocol-Interfaced Turing Machine]
\label{def:pitm}
Let $f$ denote an arbitrary computable function that a TLS implementation is asked to evaluate through the composition of its own gadgets. A TLS implementation $L$ realizes a protocol-interfaced Turing machine if there exist gadgets $G_M$, $G_A$, $G_C$, $G_I$, and $G_T$ such that
\[
L(f) = \text{output} \land \text{actuate}(\tau)
\]
where $G_M$ provides memory through Read/Write and Composition/Chaining gadgets that persist connection state, $G_A$ provides arithmetic through Arithmetic/Computation gadgets, $G_C$ provides conditionals through Control-Flow and Security-Processing gadgets, $G_I$ provides iteration through Timing/Synchronization gadgets that reenter the handshake, such as the message processing loop or the renegotiation cycle, and $G_T$ provides trust actuation through an I/O and Side-Effect gadget that converts accumulated state into an authentication or session decision $\tau$.
\end{definition}

Note that the eight gadget categories of Table~\ref{tab:tlsgadgets8} map many-to-one onto the five capabilities named in Definition~\ref{def:pitm}. Control-Flow and Security-Processing both feed $G_C$; Read/Write and Composition/Chaining both feed $G_M$. The eight categories are a taxonomy of where these capabilities appear in a TLS implementation, not five capabilities plus three extras. Table~\ref{tab:turingmapping} makes this correspondence explicit by mapping each gadget category directly onto the component of a standard Turing machine it realizes.

\begin{table}[h]
\centering
\caption{Mapping from TLS Gadget Categories to Turing Machine Components}
\label{tab:turingmapping}
\renewcommand{\arraystretch}{1.25}
\begin{tabular}{lll}
\toprule
Turing Machine Component & Supplied By (Gadget Category) & Role in $L(f)$ \\
\midrule
Tape (unbounded memory) & Read/Write, Composition/Chaining & $G_M$ \\
Symbol alphabet, arithmetic on tape contents & Arithmetic/Computation & $G_A$ \\
Transition function (state $\times$ symbol $\to$ state) & Control-Flow, Security-Processing & $G_C$ \\
Head movement, re-entry into the transition loop & Timing/Synchronization & $G_I$ \\
Halting output made externally observable & I/O and Side-Effect & $G_T$ \\
\bottomrule
\end{tabular}
\end{table}

\begin{theorem}[Protocol-Interfaced Turing Completeness]
\label{thm:tlsturing}
If a TLS implementation provides gadgets $G_M$, $G_A$, $G_C$, and $G_I$,
then the implementation can compute any computable function. If $G_T$ is
also present, the implementation can compute and simultaneously produce
a trust decision.
\end{theorem}

\begin{proof}
Table~\ref{tab:turingmapping} identifies, for each component of a standard Turing machine, the gadget category in Table~\ref{tab:tlsgadgets8} that supplies it; the proof proceeds by discharging each row of that table in turn.

\emph{Tape.} Persistent connection state accessed and modified through Read/Write and Composition/Chaining gadgets (for example, values attached to a connection through \texttt{SSL\_set\_ex\_data}) supplies unbounded memory in the same role tape cells play in a standard Turing machine: an addressable store that persists across steps and can be read and overwritten by later steps.

\emph{Symbol alphabet and arithmetic.} Arithmetic/Computation gadgets, such as digest computation and byte- or bit-length comparison, supply the arithmetic needed to transform tape contents from one step to the next, standing in for the alphabet operations a Turing machine applies at each transition.

\emph{Transition function.} Control-Flow and Security-Processing gadgets supply conditional branching, since both categories decide how execution proceeds based on a comparison against stored state. Together they realize the transition function $\delta: Q \times \Sigma \to Q \times \Sigma \times \{L, R\}$: the current handshake state and the value read from persistent memory jointly determine the next branch taken and the next value written.

\emph{Head movement and re-entry.} Timing/Synchronization gadgets that reenter the handshake, most notably a renegotiation or resumption cycle, supply iteration in the same role the read/write head's movement and the machine's repeated application of $\delta$ play in a standard Turing machine: each pass through the cycle advances the computation by one step, updating memory and re-evaluating a condition before the next pass begins.

\emph{Composing the four.} When all four gadget types are present, an implementation can simulate an arbitrary Turing machine's transition function by encoding its tape, head position, and current state in the persistent memory the $G_M$ gadgets expose, with transition rules realized as branches inside the $G_C$ gadgets that the $G_I$ gadget triggers on each pass. This is precisely the construction Table~\ref{tab:turingmapping} states row by row, and since $G_M$, $G_A$, $G_C$, and $G_I$ are each independently realizable by the corresponding gadget categories, their composition realizes the full Turing machine.

\emph{Trust actuation.} If an I/O and Side-Effect gadget is also present that converts this accumulated state into a session or authentication outcome, the same computation simultaneously produces a trust decision, establishing the $G_T$ clause.

As with the original HRU construction \cite{harrison1976protection, tripunitara2013hru}, this argument identifies each required gadget and its role, via the explicit mapping in Table~\ref{tab:turingmapping}, rather than exhibiting a machine-checked encoding of a specific universal Turing machine; Section~\ref{sec:limitations} discusses this distinction directly. Theorem~\ref{thm:turingcase} depends on the mapping established here: it applies Table~\ref{tab:turingmapping} to the sentinel and vulnerable server directly, and the demonstrations in Sections~\ref{sec:sentinel} and~\ref{sec:authbypass} instantiate this general argument concretely, exhibiting the gadgets of Table~\ref{tab:tlsgadgets8} realized as specific OpenSSL API calls and showing directly that they compose as this proof requires.
\end{proof}

The significance of this theorem is its inevitability. Any TLS
implementation that supports session resumption, certificate-based
authentication, and standard handshake processing already provides $G_M$
through $G_T$. The question is not whether a given library supports
arbitrary computation, but whether the specific gadgets needed for a
given weird machine are present in that library's implementation.

This is where OpenSSL and BoringSSL diverge in a way that matters for
construction, not just for theory. OpenSSL enables renegotiation by
default and accepts renegotiation requests from a peer without
additional configuration \cite{boringssl2024porting}. This means $G_I$
in OpenSSL includes a renegotiation-based loop that an attacker can
trigger directly, and Section~\ref{sec:testbed} grounds this claim in a
concrete implementation. BoringSSL's documentation states that it
removes this gadget entirely, since Google's engineers judged
renegotiation to be an extremely problematic feature
\cite{fastly2024boringssl}. Based on this documented design decision, we
argue that a weird machine built on renegotiation-based iteration is not
constructible against BoringSSL in the same form, though both libraries
would still satisfy Theorem~\ref{thm:tlsturing} through other gadget
combinations such as session ticket resumption. We did not build a
BoringSSL testbed to confirm this by attempting the same renegotiation
call sequence against it. Section~\ref{sec:limitations} discusses this
gap directly.

\section{Case Study and Constructive Demonstrations}
\label{sec:testbed}

To validate the theoretical results in Section~\ref{sec:theory}, we built two Docker-based testbeds using OpenSSL as the target library. Like the sentinel in Section~\ref{sec:sentinel}, the vulnerable server and attacker client are built entirely against the public OpenSSL 1.1.1w API, with no forks or patches to the library itself. Both testbeds compose the same underlying gadgets. They differ only in whether the composing application re-derives its trust decision after that decision changes.

\subsection{System Architecture}

Our testbed consists of two independent demonstrations. The first, \texttt{wm1-sentinel-demo}, implements a defensive weird machine that re-verifies session state on every read. The second, \texttt{wm2-authbypass-demo}, implements an adversarial mirror image consisting of a vulnerable server that checks state exactly once and an attacker client that exploits that gap. Figure \ref{fig:testbed} shows both architectures. Section~\ref{sec:sentinel} and Section~\ref{sec:authbypass} present each demonstration in full.

The sentinel (\texttt{main.c}) connects to a target server, captures a fingerprint of the negotiated session immediately after the handshake completes, and re-verifies that fingerprint on every subsequent read and on every renegotiation attempt. The vulnerable server (\texttt{vuln\_server.c}) checks the negotiated cipher exactly once, at initial handshake completion, and caches the result as an \texttt{approved} flag. It never re-checks that flag against the live cipher on subsequent reads. The attacker client (\texttt{attacker\_client.c}) connects with a strong cipher to pass the one-time check, then calls \texttt{SSL\_renegotiate} to downgrade the connection to a weak cipher on the same TCP connection, and reads whether the server still serves privileged data under the stale \texttt{approved} flag.

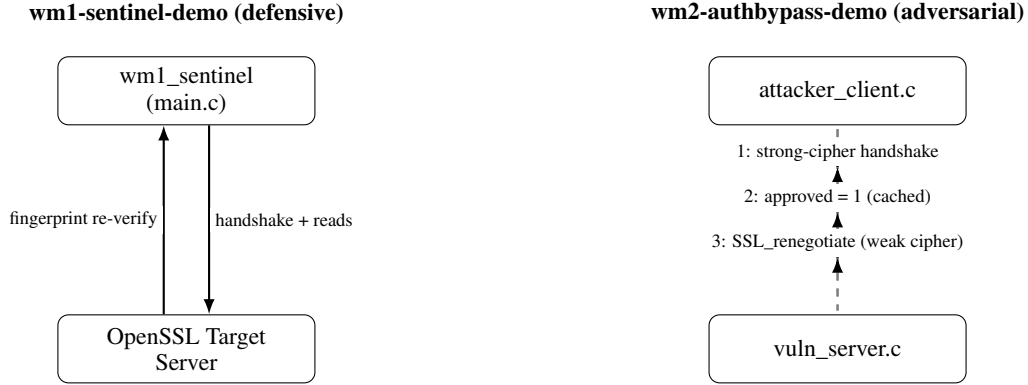
\begin{figure}[h]
\centering
\begin{tikzpicture}[
  node distance=1.2cm and 1.4cm,
  box/.style={draw, rectangle, rounded corners, minimum width=3.4cm, minimum height=0.9cm, align=center, font=\small},
  lifeline/.style={draw, dashed, thick, gray},
  msg/.style={-{Latex[length=2mm]}, thick},
  msglabel/.style={font=\scriptsize, fill=white, inner sep=2pt}
]

\node[box] (sclient) {wm1\_sentinel\\(main.c)};
\node[box, below=2.5cm of sclient] (starget) {OpenSSL Target\\Server};
\draw[msg] ([xshift=0.3cm]sclient.south) -- node[right, msglabel] {handshake + reads} ([xshift=0.3cm]starget.north);
\draw[msg] ([xshift=-0.3cm]starget.north) -- node[left, msglabel] {fingerprint re-verify} ([xshift=-0.3cm]sclient.south);

\node[above=0.3cm of sclient, font=\small\bfseries] {wm1-sentinel-demo (defensive)};

\node[box, right=5.2cm of sclient] (aattacker) {attacker\_client.c};
\node[box, right=5.2cm of starget] (avictim) {vuln\_server.c};

\coordinate (aTop) at ([yshift=-0.05cm]aattacker.south);
\coordinate (aBot) at ([yshift=0.05cm]avictim.north -| aattacker);
\coordinate (vTop) at ([yshift=-0.05cm]aattacker.south -| avictim);
\coordinate (vBot) at ([yshift=0.05cm]avictim.north);

\draw[lifeline] (aTop) -- (aBot);
\draw[lifeline] (vTop) -- (vBot);

\coordinate (a1) at ([yshift=-0.5cm]aTop);
\coordinate (v1) at ([yshift=-0.5cm]vTop);
\coordinate (a2) at ([yshift=-1.1cm]aTop);
\coordinate (v2) at ([yshift=-1.1cm]vTop);
\coordinate (a3) at ([yshift=-1.7cm]aTop);
\coordinate (v3) at ([yshift=-1.7cm]vTop);

\draw[msg] (a1) -- node[above, msglabel] {1: strong-cipher handshake} (v1);
\draw[msg] (v2) -- node[above, msglabel] {2: approved = 1 (cached)} (a2);
\draw[msg] (a3) -- node[above, msglabel] {3: SSL\_renegotiate (weak cipher)} (v3);

\node[above=0.3cm of aattacker, font=\small\bfseries] {wm2-authbypass-demo (adversarial)};

\end{tikzpicture}
\caption{TLS Weird Machines Testbed Architecture}
\label{fig:testbed}
\end{figure}

\subsection{Architectural Gadget Inventory}
\label{sec:gadgetinventory}

Table \ref{tab:gadgetcode} maps the eight gadget categories defined in Section~\ref{sec:gadgettypes} onto the actual function calls used in \texttt{main.c} and \texttt{vuln\_server.c}. The two columns are strikingly similar. The difference between defense and exploit lies almost entirely in the Timing/Synchronization and Composition/Chaining rows.

\begin{table*}[!ht]
\centering
\small
\renewcommand{\arraystretch}{1.3}
\setlength{\tabcolsep}{8pt}
\begin{tabular}{lL{6.0cm}L{6.0cm}}
\toprule
Gadget Type & Sentinel (\texttt{main.c}) & Vulnerable Server (\texttt{vuln\_server.c}) \\
\midrule
Read/Write & \texttt{SSL\_get\_peer\_certificate}, \texttt{SSL\_get\_current\_cipher}, \texttt{SSL\_read}/\texttt{SSL\_write} & \texttt{SSL\_get\_current\_cipher}, \texttt{SSL\_read}/\texttt{SSL\_write} \\
Control-Flow & State enum (\texttt{UNARMED}/\texttt{ARMED}/\texttt{RENEG}) branching in \texttt{wm1\_info\_callback} & \texttt{approved} boolean branch in \texttt{vuln\_info\_callback} and the read loop \\
Communication-Bridge & \texttt{wm1\_connect\_tcp} plus \texttt{SSL\_set\_fd} bridges raw TCP into the TLS layer & Server-side \texttt{accept}/\texttt{SSL\_set\_fd} bridges the listening socket into TLS \\
Security-Processing & SPKI pin via \texttt{EVP\_DigestFinal\_ex}, \texttt{SSL\_get\_verify\_result} & \texttt{is\_strong\_cipher} checks GCM mode and key length via \texttt{SSL\_CIPHER\_get\_bits} \\
Timing/Synchronization & \texttt{SSL\_CB\_HANDSHAKE\_START}/\texttt{DONE} hooks fire on \textbf{every} handshake and renegotiation & \texttt{SSL\_CB\_HANDSHAKE\_DONE} hook fires once meaningfully; the \texttt{else} branch logs but never re-checks \\
Arithmetic/Computation & SHA-256 digest, \texttt{memcmp} fingerprint comparison & \texttt{SSL\_CIPHER\_get\_bits} integer comparison against 256 \\
Composition/Chaining & \texttt{wm1\_SSL\_read} chains a renegotiation check, a live read, and a fingerprint re-comparison into one call & The read loop chains \texttt{SSL\_read} directly to the cached \texttt{approved} flag with no intervening re-check \\
I/O and Side-Effect & \texttt{wm1\_alert} logs and \texttt{SSL\_shutdown} terminates on mismatch under strict mode & \texttt{SSL\_write} serves \texttt{SECRET\_PRIVILEGED\_DATA} whenever \texttt{approved} is stale-true \\
\bottomrule
\end{tabular}
\caption{Architectural Gadgets Shared Between the Sentinel and the Vulnerable Server. Functions prefixed \texttt{SSL\_} or \texttt{EVP\_} are public OpenSSL API calls; \texttt{wm1\_*} and \texttt{vuln\_*} functions are application-level code written for this testbed that composes those calls.}
\label{tab:gadgetcode}
\end{table*}

This side-by-side view makes the dual-use claim concrete rather than rhetorical. Both programs use \texttt{SSL\_CTX\_set\_info\_callback}, \texttt{SSL\_get\_current\_cipher}, and \texttt{SSL\_set\_ex\_data}/\texttt{SSL\_get\_ex\_data} persistence. The sentinel composes these gadgets so that Timing/Synchronization triggers a fresh Security-Processing check on every event. The vulnerable server composes the same gadgets so that Timing/Synchronization triggers a Security-Processing check exactly once, then routes every later Read/Write gadget through a stale Control-Flow branch instead.

\subsection{Turing Completeness}
\label{sec:turingcase}

\begin{theorem}
\label{thm:turingcase}
Both the sentinel implementation and the vulnerable server described
above are Turing-complete.
\end{theorem}

\begin{proof}
Memory is provided in both programs by a persistent connection-state
struct, \texttt{wm1\_fingerprint\_t} in the sentinel and
\texttt{vuln\_conn\_state\_t} in the vulnerable server, each attached to
the connection through \texttt{SSL\_set\_ex\_data} and retrieved through
\texttt{SSL\_get\_ex\_data}. Arithmetic is provided by SHA-256 digest and
\texttt{memcmp} operations in the sentinel, and by the
\texttt{SSL\_CIPHER\_get\_bits} comparison against 256 in the vulnerable
server. Conditionals are provided by state branching in
\texttt{wm1\_info\_callback} and \texttt{vuln\_info\_callback}, both of
which decide how to proceed based on the outcome of a prior check.
Iteration is provided by each program's main read loop and by the
renegotiation path, which reenters the handshake state machine without
closing the underlying socket. Trust actuation is provided by
\texttt{SSL\_shutdown} in the sentinel and by the conditional
\texttt{SSL\_write} of \texttt{SECRET\_PRIVILEGED\_DATA} in the
vulnerable server, both of which convert accumulated state into an
externally visible outcome. By Theorem~\ref{thm:tlsturing}, both programs
can compute any computable function while simultaneously producing a
trust decision.
\end{proof}

This result establishes that Turing completeness alone says nothing about correctness. Both programs satisfy Theorem~\ref{thm:tlsturing}. Only one of them uses that capability to enforce its stated policy. Section~\ref{sec:sentinel} shows how the shared gadgets are composed for defense, and Section~\ref{sec:authbypass} shows how the same gadgets are composed for exploitation.

\section{Sentinel Demo: Beneficial Weird Machine}
\label{sec:sentinel}

The sentinel demonstrates that the same gadgets available to an attacker can be composed defensively. It is built entirely against the public OpenSSL 1.1.1w API, with no forks and no patches to the library itself.

The sentinel's design rests on a simple idea. A TLS connection is not static. Its cipher, its certificate, and its verification result can all change after the handshake completes, most commonly through renegotiation. Most applications check these properties once and never look again. The sentinel instead treats every read and every renegotiation as an opportunity to ask the same question again: does the live connection still match what was true when trust was first established.

The mechanism works in three phases. In the first phase, the sentinel captures a fingerprint immediately after the initial handshake completes, triggered by the \texttt{SSL\_CB\_HANDSHAKE\_DONE} event inside \texttt{wm1\_info\_callback}. This fingerprint records the negotiated cipher, the protocol version, a SHA-256 hash of the peer certificate's public key, and the certificate verification result. The sentinel stores this fingerprint in a struct attached to the connection through \texttt{SSL\_set\_ex\_data}, and moves its internal state from \texttt{UNARMED} to \texttt{ARMED}.

In the second phase, every call to \texttt{wm1\_SSL\_read} wraps the underlying \texttt{SSL\_read} with a fresh fingerprint capture and comparison. If the live fingerprint no longer matches the armed snapshot, the sentinel logs a violation through \texttt{wm1\_alert} and, when running in strict mode, calls \texttt{SSL\_shutdown} to terminate the connection immediately. This phase runs on every single read for the life of the connection, not just once.

In the third phase, the sentinel handles renegotiation directly. \texttt{wm1\_info\_callback} detects \texttt{SSL\_CB\_HANDSHAKE\_START} while already armed and moves its state to \texttt{RENEG}. Any read attempted during this window stalls inside \texttt{wm1\_SSL\_read} until the renegotiation handshake completes, at which point the sentinel captures a new fingerprint and compares it against the original snapshot before allowing the read to proceed. This closes exactly the gap that the vulnerable server described in Section~\ref{sec:testbed} leaves open.

We built and ran the sentinel against two scenarios inside our Docker testbed. Against a well-behaved server that never renegotiates, the sentinel completed the connection with no violations, confirming that the fingerprint comparison does not misfire under normal operation. Against a server that renegotiates down to a weaker cipher mid-connection, the same downgrade path used by \texttt{attacker\_client.c} against \texttt{vuln\_server.c}, the sentinel detected the mismatch on the first read following renegotiation and logged a violation before returning any further application data. This confirms that the same renegotiation path an attacker uses to defeat a stale policy check is directly observable to a client that re-verifies its own state, rather than trusting a single check made at the start of the connection.

The sentinel's design mirrors a lesson from the Apollo 13 mailbox. NASA's engineers did not invent a new component to solve the carbon dioxide problem. They recomposed suit hoses, plastic bags, and duct tape, materials already present on the spacecraft, into a configuration the original design never specified. The sentinel does the same with software. It introduces no new cryptographic primitive and no modification to OpenSSL itself. It only recomposes functions that already exist in the public API, \texttt{SSL\_get\_current\_cipher}, \texttt{SSL\_get\_peer\_certificate}, and \texttt{SSL\_get\_verify\_result}, into a loop that checks more often than a typical application checks. The capability to build this sentinel was present in OpenSSL the entire time. It simply required recognizing that the same gadgets an attacker chains together to defeat a stale check can be chained instead to eliminate the staleness altogether.

\section{Authbypass Demo: Adversarial Weird Machine}
\label{sec:authbypass}

The vulnerable server states a policy that sounds reasonable on its own. Privileged data should only be served over a connection using an AEAD cipher with a 256-bit key. The function \texttt{is\_strong\_cipher} implements this check directly, reading the negotiated cipher through \texttt{SSL\_get\_current\_cipher} and confirming both a GCM mode string and a bit length of at least 256. Nothing about this check is wrong. The flaw is not in what the server checks. It is in when the server checks it.

\texttt{vuln\_info\_callback} runs this check exactly once, inside the branch that fires the first time \texttt{SSL\_CB\_HANDSHAKE\_DONE} occurs. The result is cached in a per-connection struct as a boolean, \texttt{approved}, attached to the SSL object through \texttt{SSL\_set\_ex\_data}. Every subsequent read in the server's main loop trusts this cached flag without ever calling \texttt{is\_strong\_cipher} again. The callback does fire a second time if the connection renegotiates, but its \texttt{else} branch only logs the new cipher to \texttt{stderr}. It never re-derives the policy decision. This resembles a known class of anti-pattern in TLS deployments: a security-relevant property is checked once at connection setup, and the application's logic assumes it remains valid for the life of the connection, even though TLS explicitly allows that property to change mid-connection through renegotiation.

The attacker exploits this gap using nothing but the public OpenSSL client API, in a sequence \texttt{attacker\_client.c} carries out in three steps. First, the attacker connects honestly. \texttt{SSL\_set\_cipher\_list} is set to \texttt{ECDHE-RSA-AES256-GCM-SHA384} before the initial handshake, a cipher that satisfies the server's policy. \texttt{SSL\_connect} completes, \texttt{vuln\_info\_callback} runs its one-time check, \texttt{is\_strong\_cipher} returns true, and \texttt{approved} is set to 1. At this point the server's cached decision and the live connection state agree, and the attacker's first request for \texttt{GET\_SECRET} returns nothing of value, since the server has not yet been asked to serve the actual privileged payload in this exchange.

Second, the attacker renegotiates. The attacker calls \texttt{SSL\_set\_cipher\_list} again, this time with \texttt{AES128-SHA}, a cipher with no forward secrecy and a 128-bit key that would fail \texttt{is\_strong\_cipher} outright. The attacker then calls \texttt{SSL\_renegotiate} followed by \texttt{SSL\_do\_handshake}, which reenters the TLS state machine on the same TCP connection and completes a full new handshake under the weaker cipher. OpenSSL permits this by design. Nothing in the protocol or the library prevents a peer from renegotiating to a weaker cipher suite than the one originally negotiated, and the server's \texttt{SSL\_CTX\_set\_cipher\_list} was deliberately configured with \texttt{ALL:@SECLEVEL=0} to allow exactly this range, mirroring real-world servers that support broad cipher ranges for backward compatibility.

Third, the attacker collects the reward. \texttt{vuln\_info\_callback} fires again on the completed renegotiation and logs the new cipher, but the cached \texttt{approved} flag from the first handshake is never touched. The server's main loop continues to route every \texttt{SSL\_read} through that stale flag. The attacker sends the same \texttt{GET\_SECRET} request it sent before the downgrade, and the server responds with \texttt{SECRET\_PRIVILEGED\_DATA=42}, now traveling over a connection that no longer satisfies the server's own stated policy.

We built and ran this exploit end to end inside our Docker testbed. The attacker's own diagnostic output confirms each stage of the attack directly: it logs the cipher negotiated on the initial handshake, logs the renegotiation to a weaker cipher, and checks the server's response for the string \texttt{SECRET\_PRIVILEGED\_DATA} to confirm whether the bypass succeeded. In our runs, the server served the privileged payload after the downgrade in every trial, since nothing in \texttt{vuln\_server.c} ever invalidates the cached decision once it is set. We also ran the sentinel from Section~\ref{sec:sentinel} against this same downgrade path as a point of direct comparison, and it detected the mismatch on the first post-renegotiation read, confirming that the vulnerable server's failure is a composition choice, not a limitation of what OpenSSL exposes to the application.

This exploit requires no memory corruption, no malformed input, and no code outside the public OpenSSL API. The attacker calls four functions, \texttt{SSL\_set\_cipher\_list}, \texttt{SSL\_renegotiate}, \texttt{SSL\_do\_handshake}, and \texttt{SSL\_read}, in a sequence the server's author did not anticipate. Every one of these calls is also available to, and used by, the sentinel in Section~\ref{sec:sentinel}. The vulnerability is not a defect in any single function. It is the absence of one composition step: a fresh Security-Processing check chained to the Timing/Synchronization event of a completed renegotiation. The server has that event available. It simply does not act on it a second time.

\section{Dual-Use Analysis}
\label{sec:dualuse}

Sections \ref{sec:sentinel} and \ref{sec:authbypass} demonstrated that the sentinel and the vulnerable server draw from the same gadget inventory. This section makes that symmetry explicit, then argues that the symmetry is not incidental. It is a structural property of any TLS implementation that supports renegotiation, and removing it carries a cost.

\subsection{Gadget Symmetry}

Figure \ref{fig:statemachines} makes the divergence between these two programs explicit as a pair of state machines. The sentinel's machine, on the left, forms a cycle: every renegotiation forces a return trip through \texttt{ARMED} before the connection can produce output again, and a mismatch routes to termination rather than continuing silently. The vulnerable server's machine, on the right, has no such cycle. Once it reaches \texttt{STALE}, no path leads back to a fresh check. The server continues to serve privileged data from that state indefinitely, regardless of how many times the underlying cipher changes.

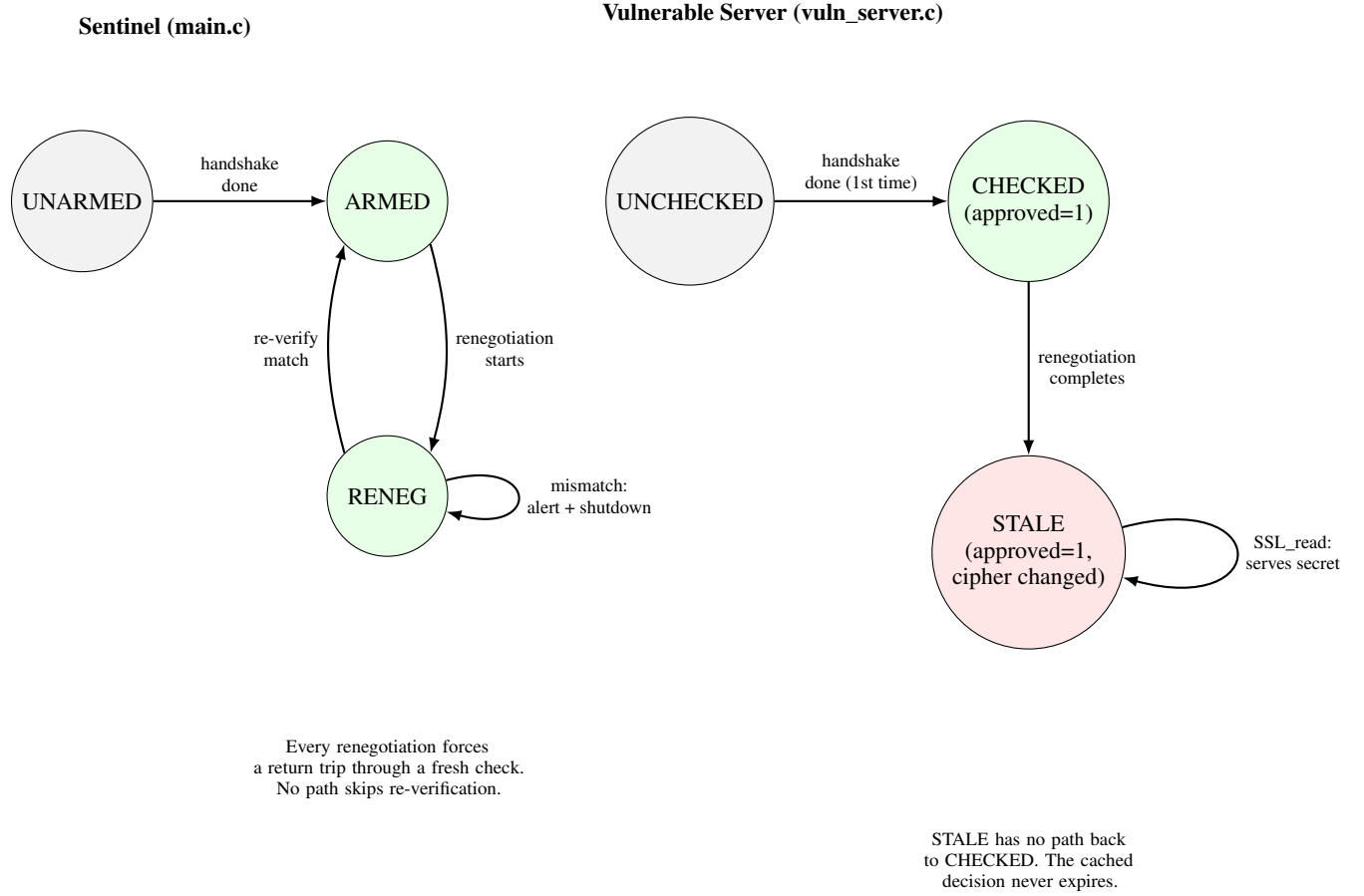
\begin{figure}[h]
\centering
\begin{tikzpicture}[
  node distance=2.3cm,
  state/.style={draw, circle, minimum size=1.6cm, align=center, font=\small},
  goodstate/.style={state, fill=green!10},
  badstate/.style={state, fill=red!10},
  neutralstate/.style={state, fill=gray!10},
  tr/.style={-{Latex[length=2mm]}, thick},
  trlabel/.style={font=\scriptsize, align=center}
]

\node[neutralstate] (unarmed) {UNARMED};
\node[goodstate, right=of unarmed] (armed) {ARMED};
\node[goodstate, below=of armed] (reneg) {RENEG};

\draw[tr] (unarmed) -- node[trlabel, above] {handshake\\done} (armed);
\draw[tr] (armed.south east) to[bend left=15] node[trlabel, right] {renegotiation\\starts} (reneg.north east);
\draw[tr] (reneg.north west) to[bend left=15] node[trlabel, left] {re-verify\\match} (armed.south west);
\draw[tr] (reneg) to[loop right] node[trlabel, right] {mismatch:\\alert + shutdown} (reneg);

\node[above=1.1cm of unarmed, xshift=1.1cm, font=\small\bfseries] {Sentinel (main.c)};
\node[below=2.3cm of reneg, font=\scriptsize, align=center, text width=4.5cm]
  {Every renegotiation forces\\a return trip through a fresh check.\\No path skips re-verification.};

\node[neutralstate, right=6cm of unarmed] (unchecked) {UNCHECKED};
\node[goodstate, right=of unchecked] (checked) {CHECKED\\(approved=1)};
\node[badstate, below=of checked] (stale) {STALE\\(approved=1,\\cipher changed)};

\draw[tr] (unchecked) -- node[trlabel, above] {handshake\\done (1st time)} (checked);
\draw[tr] (checked) -- node[trlabel, right] {renegotiation\\completes} (stale);
\draw[tr] (stale) to[loop right] node[trlabel, right] {SSL\_read:\\serves secret} (stale);

\node[above=1.1cm of unchecked, xshift=1.1cm, font=\small\bfseries] {Vulnerable Server (vuln\_server.c)};
\node[below=2.3cm of stale, font=\scriptsize, align=center, text width=4.5cm]
  {STALE has no path back\\to CHECKED. The cached\\decision never expires.};

\end{tikzpicture}
\caption{State Machine Comparison: Sentinel Re-Verification vs. Vulnerable Server's One-Shot Check}
\label{fig:statemachines}
\end{figure}

Table \ref{tab:gadgetcode} in Section~\ref{sec:testbed} already showed that the sentinel and the vulnerable server use nearly identical gadgets. Both rely on \texttt{SSL\_CTX\_set\_info\_callback} for Timing/Synchronization, both call \texttt{SSL\_get\_current\_cipher} for Security-Processing, and both persist state across reads using \texttt{SSL\_set\_ex\_data} and \texttt{SSL\_get\_ex\_data}. The single point of divergence is Composition/Chaining. The sentinel chains a fresh Security-Processing check to every Timing/Synchronization event. The vulnerable server chains a Security-Processing check to only the first such event, then routes every later Read/Write gadget through a stale Control-Flow branch instead.

This resembles a similar class of pattern documented in prior TLS state machine research. Beurdouche et al. found that composite state machines allowed message sequences the protocol specification never intended, producing authentication bypasses across multiple libraries without any memory corruption \cite{beurdouche2015smack}. CVE-2020-2655 showed a structurally comparable issue in a different implementation, where reordering handshake messages skipped client authentication entirely \cite{cve2020_2655}. In each of these cases, including our own, the vulnerability is not a broken gadget. It is a missing composition step between two gadgets that both function exactly as designed.

The practical consequence is that no single gadget in our taxonomy can be labeled safe or unsafe in isolation. Renegotiation is not inherently dangerous. A sentinel that re-verifies on renegotiation and a vulnerable server that does not are built from the same renegotiation gadget. What determines the outcome is entirely how many times, and at which events, the application chains a Security-Processing check to the rest of the connection lifecycle.

\subsection{The Cost of Elimination}

One response to this symmetry is to remove the gadget entirely. If renegotiation enables both the sentinel's defensive re-verification and the attacker's downgrade, removing renegotiation from a TLS library would appear to close the vulnerability at its root. This is close to the position BoringSSL has taken. BoringSSL disables renegotiation by default, and Google's engineers describe it as an extremely problematic feature of the protocol \cite{boringssl2024porting, fastly2024boringssl}.

This response works, but it is not free. The same renegotiation path that
the attacker used to downgrade a cipher is also the path a legitimate
application would use to step up authentication mid-connection, for
example to request a client certificate only after a user reaches a page
that requires it, without tearing down and rebuilding the entire
connection. Removing renegotiation removes this capability along with
the vulnerability. This mirrors the argument we made about metastability
in Section~\ref{sec:theory}. A flip-flop cannot be redesigned to
eliminate its metastable state without changing what a flip-flop is. A
TLS library cannot remove every gadget capable of adversarial
composition without also removing the gadgets that legitimate
applications depend on for flexibility. The choice is not between a
secure library and an insecure one. It is between a library that keeps a
capability and manages its risk through careful application design, and
a library that removes the capability and accepts the loss of whatever
legitimate use case depended on it.

Lamport's own extension of this argument sharpens the point further.
Buridan's Principle states that any mechanism built to eliminate a
discrete decision's indecision does not remove the underlying
possibility of indecision. It only produces a new mechanism of the same
general type, one that still admits inputs capable of driving it into an
arbitrarily long or ambiguous decision time, even if those inputs are
no longer the same ones that afflicted the original design
\cite{lamport1984buridan}. Removing renegotiation from BoringSSL is this
same move applied to a TLS library. The specific renegotiation-based
gadget is gone, but the library still provides memory, arithmetic,
conditionals, and iteration through the gadgets it retains, such as
session ticket resumption, and Lamport's argument gives us reason to
expect that some composition of those remaining gadgets admits an
analogous weird machine. The difference is not that the underlying
possibility has been eliminated. It is that the specific inputs
capable of exploiting it are no longer the ones this paper has
identified, and, following Lamport's own point about the RS flip-flop,
there is no general procedure for locating them in advance. BoringSSL's
design choice narrows where the risk can occur. It does not certify
that no such gadget exists in what remains.

\subsection{Library-Level Consequences}

OpenSSL and BoringSSL made different choices along this tradeoff, and our testbed shows the practical effect of that difference on OpenSSL's side. OpenSSL enables renegotiation by default and accepts a peer's renegotiation request without additional configuration \cite{boringssl2024porting}. This is precisely the behavior \texttt{vuln\_server.c} relies on. The server's permissive cipher list, set through \texttt{SSL\_CTX\_set\_cipher\_list(ctx, "ALL:@SECLEVEL=0")}, combined with OpenSSL's default acceptance of renegotiation, is what allows \texttt{attacker\_client.c} to call \texttt{SSL\_renegotiate} and succeed without the server rejecting the request outright. The exploit in Section~\ref{sec:authbypass} is constructible in OpenSSL specifically because OpenSSL keeps this gadget available and low-friction to invoke.

BoringSSL's documentation indicates that it closes this path differently than an application-level fix would. Rather than requiring every application built on BoringSSL to remember to re-verify state after a renegotiation, BoringSSL removes the renegotiation gadget itself for most connection types \cite{boringssl2024porting, fastly2024boringssl}. Based on this documented removal, we argue that an application built against BoringSSL would not be able to reproduce the same attack in the same form, not because BoringSSL contains better cipher-strength checking logic, but because the Timing/Synchronization gadget the attacker depends on is not present in the same form. As noted in Section~\ref{sec:limitations}, we did not build a BoringSSL testbed to confirm this directly, so this claim rests on BoringSSL's documented design decision rather than on an attempted reproduction.

This has a direct implication for how the two libraries should be evaluated. A security review of an OpenSSL-based application cannot stop at checking whether cryptographic primitives are implemented correctly. It must also check whether every security-relevant decision is re-derived after every event that can legally change the state that decision depended on, since OpenSSL will not prevent an application from making that mistake. A security review of a BoringSSL-based application starts from a narrower set of gadgets, which reduces the number of places a composition error like the one in \texttt{vuln\_server.c} could occur, but it does not eliminate the general risk. Any gadget BoringSSL does retain, such as session ticket resumption, remains available for exactly the same kind of adversarial composition, and an application that caches a decision across a resumed session without re-checking it would be vulnerable to a variant of the same attack, built from a different gadget in the same taxonomy.

The dual-use property we have described is therefore not a defect that either library can patch away. It is a consequence of providing memory, arithmetic, conditionals, and iteration in a protocol state machine, as Theorem~\ref{thm:tlsturing} established. OpenSSL and BoringSSL differ in which specific gadgets they expose, and that difference changes which weird machines are constructible in each library. Neither library can offer the full flexibility TLS applications sometimes need while also removing every gadget capable of adversarial composition. Library maintainers can only choose where along that tradeoff their default configuration sits, and application developers must know which gadgets their chosen library retains.

\section{Limitations and Future Work}
\label{sec:limitations}

The results in Sections~\ref{sec:theory}--\ref{sec:dualuse} establish that
TLS weird machines are constructible and that OpenSSL and BoringSSL differ
in which gadgets they expose. This section states plainly what those
results do not establish, and what closing each gap would require. We
take this stance deliberately: a claim of dual-use architectural risk is
only useful to the community if its boundaries are stated as precisely as
the claim itself.

\subsection{BoringSSL is argued, not demonstrated}
\label{sec:lim-boringssl}

Our comparison between OpenSSL and BoringSSL in Section~\ref{sec:dualuse}
is a documentation-level argument. We show that BoringSSL's porting guide
and public engineering statements describe renegotiation as disabled by
default \cite{boringssl2024porting, fastly2024boringssl}, and we reason
from that description to the conclusion that the specific authentication
bypass in Section~\ref{sec:authbypass} is not constructible against a
BoringSSL-linked version of \texttt{vuln\_server.c}. We did not build a
BoringSSL testbed, port \texttt{attacker\_client.c} against
\texttt{libssl} from BoringSSL, or attempt the renegotiation call sequence
against a real BoringSSL server to confirm that it is rejected, fails
silently, or fails in some other way than we assume. It is possible that
BoringSSL exposes a different, narrower renegotiation-like gadget, for
example through session tickets or a restricted post-handshake
authentication path, that supports a variant of the same attack under
different function names. Closing this gap requires building a second
Docker testbed with BoringSSL in place of OpenSSL, attempting the
identical three-step exploit from Section~\ref{sec:authbypass} against it,
and reporting the concrete failure mode rather than inferring one from
documentation. Until that testbed exists, the library-level comparison in
Section~\ref{sec:dualuse} should be read as a documented hypothesis, not
an empirical result on the same footing as the OpenSSL demonstrations.

\subsection{Testbed scope: EOL library, localhost, no network adversary}
\label{sec:lim-testbed}

Both demonstrations run against OpenSSL 1.1.1w, a version that has passed
its end-of-life date and no longer receives security patches. We chose
this version because it retains the legacy renegotiation API our gadgets
depend on, but this means our results describe a TLS library version most
production deployments should no longer be running. We did not verify
whether the same gadget inventory and the same missing composition step
persist in an actively maintained OpenSSL 3.x branch, where renegotiation
handling has been partially restructured. A finding that holds only on a
deprecated branch is a weaker claim than one that holds on the currently
maintained line, and a fuller treatment would repeat the gadget inventory
in Table~\ref{tab:gadgetcode} against OpenSSL 3.x before drawing
conclusions about current deployments.

Separately, both the sentinel and the authentication bypass run inside a
single Docker Compose network, with client and server as adjacent
containers exchanging traffic over localhost-equivalent virtual
interfaces. There is no network path with realistic latency, packet loss,
middlebox interference, or an on-path or off-path network adversary
distinct from the TLS peer itself. This matters because several real-world
constraints that would affect exploit reliability, such as TCP
retransmission interacting with a mid-connection renegotiation, or a
middlebox that strips or rewrites renegotiation records, cannot surface in
a Docker-internal network. A fuller evaluation would repeat the
authentication bypass against physically or virtually separated hosts,
across an actual network path, and would report whether the same
three-call sequence still succeeds without modification.

\subsection{Theorems 1 and 2 are proof sketches}
\label{sec:lim-proofs}

The proofs given for Theorem~\ref{thm:tlsturing} and Theorem~\ref{thm:turingcase}
follow the same style as the original HRU undecidability argument
\cite{harrison1976protection, bishop2002computer}. They identify a gadget
in the implementation for each required primitive, memory, arithmetic,
conditionals, and iteration, and argue informally that composing them
yields arbitrary computation. This is standard practice in the weird
machine literature \cite{bratus2011exploit, domas2013mov,
shapiro2013weird}, but it falls short of a machine-checked or fully
formalized proof. In particular, we do not exhibit an explicit encoding of
an arbitrary Turing machine's transition function into the specific
gadgets of \texttt{main.c} and \texttt{vuln\_server.c}, in the way Domas's
MOV result gives a complete instruction-translation table
\cite{domas2013mov} or the way HRU's original proof gives an explicit
command-to-transition mapping \cite{harrison1976protection}. Our claim is
that the necessary primitives are present and composable, which is
sufficient to support the paper's central argument about architectural
inevitability, but it is weaker than a constructive completeness proof.
Closing this gap, in the spirit of Dullien's more formal treatment of
weird machine exploitability \cite{dullien2018weird}, would mean encoding
a specific universal Turing machine directly into the gadget set of
Table~\ref{tab:tlsgadgets8} and verifying the encoding mechanically rather
than by inspection.

\subsection{No adaptive or red-team attacker}
\label{sec:lim-adaptive}

The authentication bypass in Section~\ref{sec:authbypass} is a single,
fixed, scripted exploit: \texttt{attacker\_client.c} always negotiates
honestly, always renegotiates down to the same weak cipher, and always
issues the same follow-up request. We did not build an adaptive attacker
that varies its strategy based on server responses, probes for which
composition errors are present before committing to an attack path, or
attempts to defeat the sentinel from Section~\ref{sec:sentinel} directly.
This is the same category of gap as in prior TLS state-machine work
\cite{beurdouche2015smack, cve2020_2655}, where the reported flaw is a
single reproducible message sequence rather than the output of a search
over the full space of reachable states. A more demanding evaluation would
implement a red-team process that treats the gadget taxonomy in
Table~\ref{tab:tlsgadgets8} as a search space, systematically attempting
compositions the vulnerable server's author did not anticipate, and would
also attempt to evade or defeat the sentinel's fingerprint re-verification
rather than assuming it holds. Building this adversary is a natural next
step once the base gadget taxonomy and the single-path demonstrations here
are taken as validated, but it is a distinct undertaking from what this
paper attempts.

\subsection{Single vulnerable pattern, single protocol version}
\label{sec:lim-scope}

\texttt{vuln\_server.c} encodes exactly one composition error: a
cipher-strength check performed once at initial handshake and never
re-derived after renegotiation. Section~\ref{sec:dualuse} argues, but does
not demonstrate, that other gadgets in Table~\ref{tab:tlsgadgets8}, such as
session ticket resumption, support variants of the same attack pattern.
We did not build a second vulnerable server around session-ticket reuse or
around extension-parsing state to confirm this generalization empirically.
Our theoretical framework is also stated and evaluated only against TLS
1.2-style renegotiation; TLS 1.3 removed renegotiation as originally
specified and replaced it with post-handshake authentication
\cite{rfc8446}, and we did not evaluate whether an analogous composition
error is constructible against that mechanism. Extending the case study to
a second vulnerable pattern and to TLS 1.3's post-handshake authentication
path would test whether the architectural-inevitability argument in
Section~\ref{sec:theory} generalizes beyond the single mechanism
demonstrated here, or whether some TLS gadgets are more resistant to this
class of composition error than others.

\subsection{Constraining composition order as a library-level defense}
\label{sec:lim-orderconstraint}

This paper follows the weird machine literature in treating the
attacker's input as the mechanism of exploitation: an adversary supplies a
sequence of otherwise-legitimate API calls in an order the
implementation's author did not anticipate \cite{bratus2011exploit}. Our
framing adds a second axis to this picture. The authentication bypass in
Section~\ref{sec:authbypass} is not attacker-supplied input in the
traditional sense; \texttt{attacker\_client.c} calls only functions the
OpenSSL API already exposes, in an order the API itself permits. The
vulnerability is therefore not just a property of what the attacker
supplies, but of what the library's public interface allows to be
supplied in what order. Section~\ref{sec:dualuse} frames this as a
Composition/Chaining decision made by the application, not the library,
but that framing leaves open whether the library itself could take on
part of that responsibility.

This suggests a defense we do not evaluate here: rather than removing a
gadget outright, as BoringSSL does with renegotiation
(Section~\ref{sec:lim-boringssl}), a library could instead constrain the
\emph{permitted order} of gadget invocation at the API boundary, rejecting
call sequences that reach a Security-Processing decision without a
subsequent re-check, in the same way a protocol conformance checker
rejects out-of-order handshake messages. Such a mechanism would preserve
the underlying capability, such as renegotiation, while narrowing which
compositions of it are reachable through the public API. Concretely, this
would mean the library tracking, per connection, whether a
Security-Processing gadget has been invoked since the most recent
Timing/Synchronization event, and refusing to serve data through a
Read/Write gadget until it has. We did not design or evaluate an
order-constraint mechanism of this kind; doing so would require
specifying, for each gadget category in Table~\ref{tab:tlsgadgets8}, which
sequences of calls are safe to permit, which is a nontrivial design
problem in its own right, since an overly strict ordering constraint risks
reintroducing the same loss of legitimate flexibility discussed in
Section~\ref{sec:dualuse}. We leave this as a natural next step beyond the
taxonomy this paper establishes, distinct from both the gadget-removal
approach BoringSSL takes and the application-level discipline OpenSSL
currently requires.

None of these gaps affect the paper's central claim: the sentinel and the
vulnerable server are built from the same public OpenSSL gadgets, and the
difference between defense and exploit is a missing composition step
rather than a memory-safety bug.

\section{Conclusion}
\label{sec:conclusion}

This paper extended weird machine theory to network security protocols, using the TLS handshake as
implemented in OpenSSL and BoringSSL as a case study. We made four
contributions. First, we formalized trust actuation as the protocol-level
analog to physical actuation, showing that a TLS implementation's
computation is coupled not to a motor or a relay but to an authentication
decision that persists once made. Second, we defined a gadget taxonomy
specific to TLS, covering handshake state, session storage, renegotiation
and resumption, security processing, and the other categories detailed in
Section~\ref{sec:theory}, and argued, following the same proof-sketch
style as the classical HRU result, that any implementation providing
these gadgets satisfies the conditions for Turing completeness. Third, we
built and validated two working demonstrations against real OpenSSL code
paths, a sentinel that re-verifies session state on every read and an
authentication bypass that exploits a server's failure to do the same,
showing directly that the same gadgets support both outcomes. Fourth, we
compared OpenSSL and BoringSSL and argued that library design choices,
not implementation bugs alone, determine which weird machines are
constructible in each, a claim we demonstrated directly for OpenSSL and
inferred from documented design decisions for BoringSSL.

The central finding of this paper is that Turing completeness in a TLS
implementation is not a defect to be found and patched. It is a necessary
consequence of providing the memory, arithmetic, conditionals, and
iteration that session resumption and renegotiation require. The
vulnerable server in Section~\ref{sec:authbypass} contains no bug in the
traditional sense. Every function it calls behaves exactly as OpenSSL
specifies. Its failure is a missing composition step, one point where a
cached decision should have been re-derived and was not. This mirrors the
lesson from Lamport's metastable flip-flop and from the HRU access
control result. The capability was never absent from the specification
by oversight. It was present because the underlying primitives make it
unavoidable.

This finding carries a direct implication for how TLS libraries should be
designed and evaluated. A library cannot remove every gadget capable of
adversarial composition without also removing the flexibility that
legitimate applications depend on. Library maintainers should therefore treat
gadget inventory as a first-class design decision, not an afterthought to
be discovered later through vulnerability disclosure. Documenting which
gadgets a library exposes, and under what conditions those gadgets can be
chained together, would let application developers reason about
composition risk the same way they already reason about cryptographic
primitive choice. OpenSSL's approach of retaining broad functionality
places this burden on the application. BoringSSL's approach of removing
entire gadget classes places the burden on the library, at the cost of
flexibility some applications require. Neither choice eliminates the
underlying property. It is a general characteristic of any protocol
implementation with sufficient state, and future library design should
treat it as such rather than as a series of individual bugs to be fixed
one at a time.

\section*{Code Availability}

The code for the two demonstrations described in this paper, \texttt{wm1-sentinel-demo} and \texttt{wm2-authbypass-demo}, is publicly available at \url{https://github.com/rossgore/weird_machine_gadgets/tree/main/tls-weird-machines}. The repository includes the sentinel implementation (\texttt{main.c}), the vulnerable server (\texttt{vuln\_server.c}), the attacker client (\texttt{attacker\_client.c}), and the Docker Compose configuration used to build and run both testbeds described in Section~\ref{sec:testbed}.

\section*{Author Contributions}

The following author contributions are categorized in Table \ref{tab:author_contrib} according to the \textit{CRediT (Contributor Roles Taxonomy)} \cite{credit2022contributor}. The author order in this paper is strictly alphabetical and does not imply relative levels of contribution.

\begin{table*}[!ht]
\centering
\renewcommand{\arraystretch}{1.2}
\resizebox{\textwidth}{!}{%
\begin{tabular}{lccccccccccccc}
\toprule
\textbf{Author} &
\rotatebox[origin=c]{60}{\parbox{3.2cm}{\centering Conceptualization}} &
\rotatebox[origin=c]{60}{\parbox{3.2cm}{\centering Formal Analysis}} &
\rotatebox[origin=c]{60}{\parbox{3.2cm}{\centering Investigation}} &
\rotatebox[origin=c]{60}{\parbox{3.2cm}{\centering Methodology}} &
\rotatebox[origin=c]{60}{\parbox{3.2cm}{\centering Resources}} &
\rotatebox[origin=c]{60}{\parbox{3.2cm}{\centering Software}} &
\rotatebox[origin=c]{60}{\parbox{3.2cm}{\centering Visualization}} &
\rotatebox[origin=c]{60}{\parbox{3.2cm}{\centering Writing – Original Draft}} &
\rotatebox[origin=c]{60}{\parbox{3.2cm}{\centering Writing – Review \& Editing}} &
\rotatebox[origin=c]{60}{\parbox{3.2cm}{\centering Funding Acquisition}} &
\rotatebox[origin=c]{60}{\parbox{3.2cm}{\centering Project Administration}} &
\rotatebox[origin=c]{60}{\parbox{3.2cm}{\centering Supervision}} &
\rotatebox[origin=c]{60}{\parbox{3.2cm}{\centering Validation}} \\
\midrule
Michael Collins        & X &   & X & X &   &   &   &   & X &   &   & X &   \\
Jada Cumberland        &   &   & X & X & X & X & X & X &   &   &   &   &   \\
Brianne Dunn           &   &   & X & X & X & X & X & X &   &   &   &   &   \\
Ross Gore              & X &   &   & X &   &   &   &   & X &   & X & X &   \\
Samuel Jackson         &   &   & X & X & X & X & X & X &   &   &   &   &   \\
Sachin Shetty          &   &   &   &   &   &   &   &   &   & X & X & X &   \\
Jonathan Takeshita     &   & X &   & X &   &   &   &   &   &   &   &   & X \\
\bottomrule
\end{tabular}%
}
\caption{Author contributions per CRediT taxonomy (X indicates contribution). \label{tab:author_contrib}}
\end{table*}

\section*{Acknowledgments}

This work was supported by an INSuRE+C AY 25-26 grant through the National Center of Academic Excellence in Cybersecurity (NCAE). Computational resources were provided by Old Dominion University.

\bibliographystyle{unsrt}  
\bibliography{references}

@misc{rfc8446,
author = {Eric Rescorla},
title = {The Transport Layer Security ({TLS}) Protocol Version 1.3},
howpublished = {IETF RFC 8446},
year = {2018},
url = {https://www.rfc-editor.org/rfc/rfc8446}
}

@article{tripunitara2013hru,
  author = {Mahesh V. Tripunitara and Ninghui Li},
  title = {The Foundational Work of Harrison-Ruzzo-Ullman Revisited},
  journal = {IEEE Transactions on Dependable and Secure Computing},
  volume = {10},
  number = {1},
  pages = {28--39},
  year = {2013}
}

@misc{heartbleed2014,
author = {{MITRE}},
title = {{CVE-2014-0160}: {OpenSSL} Heartbeat Extension Buffer Over-Read},
year = {2014},
howpublished = {NIST NVD},
url = {https://nvd.nist.gov/vuln/detail/CVE-2014-0160}
}

@misc{stackademic2024differences,
author = {{Stackademic}},
title = {{OpenSSL} vs {BoringSSL}: Understanding the Differences},
year = {2024},
howpublished = {Stackademic Blog},
url = {https://blog.stackademic.com/openssl-vs-boringssl-understanding-the-differences-534764cdb24e}
}

@misc{cossacklabs2017replacing,
author = {{Cossack Labs}},
title = {Replacing {OpenSSL} with {BoringSSL} in a Complex Multi-Platform Application},
year = {2017},
howpublished = {Cossack Labs Blog},
url = {https://www.cossacklabs.com/blog/replacing-openssl-with-boringssl/}
}

@inproceedings{beurdouche2015smack,
author = {Benjamin Beurdouche and Karthikeyan Bhargavan and Antoine Delignat-Lavaud and Cedric Fournet and Markulf Kohlweiss and Alfredo Pironti and Pierre-Yves Strub and Jean Karim Zinzindohoue},
title = {A Messy State of the Union: Taming the Composite State Machines of {TLS}},
booktitle = {IEEE Symposium on Security and Privacy (S\&P)},
year = {2015},
pages = {535--552},
doi = {10.1109/SP.2015.39}
}

@misc{cve2020_2655,
author = {{NIST National Vulnerability Database}},
title = {{CVE-2020-2655}: Java {SE} {JSSE} Client Authentication Bypass},
year = {2020},
howpublished = {NIST NVD},
url = {https://nvd.nist.gov/vuln/detail/CVE-2020-2655}
}

@misc{boringssl2024porting,
author = {{Google BoringSSL Project}},
title = {Porting from {OpenSSL} to {BoringSSL}},
year = {2024},
howpublished = {BoringSSL Documentation},
url = {https://boringssl.googlesource.com/boringssl/+/main/PORTING.md}
}

@misc{fastly2024boringssl,
author = {{Fastly}},
title = {{BoringSSL} to Make {TLS} More Secure},
year = {2024},
howpublished = {Fastly Blog},
url = {https://www.fastly.com/blog/boringssl-to-make-tls-more-secure}
}

@article{bratus2011exploit,
  author  = {Sergey Bratus and Meredith E. Locasto and Michael L. Patterson and Len Sassaman and Anna Shubina},
  title   = {Exploit Programming: From Buffer Overflows to Weird Machines and Theory of Computation},
  journal = {;login: The USENIX Magazine},
  volume  = {36},
  number  = {6},
  pages   = {13--21},
  year    = {2011},
}

@article{lamport1984buridan,
  author  = {Leslie Lamport},
  title   = {Buridan's Principle},
  journal = {Foundations of Physics},
  volume  = {14},
  number  = {11},
  pages   = {1111--1122},
  year    = {1984},
  doi     = {10.1007/BF00728853},
}

@article{lamport1985interprocess,
  author  = {Leslie Lamport},
  title   = {On Interprocess Communication},
  journal = {Distributed Computing},
  volume  = {1},
  number  = {2},
  pages   = {77--101},
  year    = {1986},
  doi     = {10.1007/BF01786228},
}

@misc{xbox360glitch,
  author       = {{Free60 Project}},
  title        = {Reset Glitch Hack: Technical Documentation},
  year         = {2011},
  howpublished = {Xbox Hacking Community},
  url          = {https://free60.org/Reset_Glitch_Hack},
  note         = {Voltage glitching to bypass secure boot},
}

@inproceedings{domas2013mov,
  author    = {Christopher Domas},
  title     = {M/o/Vfuscator: Turning {MOV} into a Soul-Crushing Compiler},
  booktitle = {DEF CON 23},
  year      = {2013},
}

@inproceedings{shapiro2013weird,
  author    = {Rebecca Shapiro and Sergey Bratus and Sean W. Smith},
  title     = {Weird Machines in {ELF}: A Spotlight on the Underappreciated Metadata},
  booktitle = {USENIX Workshop on Offensive Technologies (WOOT)},
  year      = {2013},
}

@article{credit2022contributor,
  title={Contributor roles taxonomy},
  author={CRediT, CASRAI},
  journal={URL: https://credit.niso.org},
  year={2022}
}

@inproceedings{bangert2013printable,
  author    = {Julian Bangert and Sergey Bratus and Rebecca Shapiro and Sean W. Smith},
  title     = {The Page-Fault Weird Machine: Lessons in Instruction-Less Computation},
  booktitle = {USENIX Workshop on Offensive Technologies (WOOT)},
  year      = {2013},
}

@article{dullien2018weird,
  author  = {Thomas Dullien},
  title   = {Weird Machines, Exploitability, and Provable Unexploitability},
  journal = {IEEE Transactions on Emerging Topics in Computing},
  volume  = {8},
  number  = {2},
  pages   = {391--403},
  year    = {2018},
  doi     = {10.1109/TETC.2017.2785299},
}

@inproceedings{murdock2021weird,
  author    = {Kit Murdock and David Oswald and Flavio D. Garcia and Jo Van Bulck and Frank Piessens and Daniel Gruss},
  title     = {Plundervolt: Software-Based Fault Injection Attacks Against Intel {SGX}},
  booktitle = {IEEE Symposium on Security and Privacy (S\&P)},
  year      = {2020},
  pages     = {1466--1482},
  doi       = {10.1109/SP40000.2020.00057},
}

@book{edge2020apollo,
  title={Apollo 13: a successful failure},
  author={Edge, Laura B},
  year={2020},
  publisher={Twenty-First Century Books (Tm)}
}

@article{harrison1976protection,
  author  = {Michael A. Harrison and Walter L. Ruzzo and Jeffrey D. Ullman},
  title   = {Protection in Operating Systems},
  journal = {Communications of the ACM},
  volume  = {19},
  number  = {8},
  pages   = {461--471},
  year    = {1976},
  doi     = {10.1145/360303.360333},
}

@book{bishop2002computer,
  author    = {Matt Bishop},
  title     = {Computer Security: Art and Science},
  publisher = {Addison-Wesley},
  year      = {2002},
  isbn      = {978-0201440997},
}

@book{lovell1994lost,
  author    = {Jim Lovell and Jeffrey Kluger},
  title     = {Lost Moon: The Perilous Voyage of Apollo 13},
  publisher = {Houghton Mifflin},
  year      = {1994},
  isbn      = {978-0395670293},
}

@misc{nasa2008mailbox,
  author       = {{NASA}},
  title        = {Apollo 13 Lunar Module 'Mail Box'},
  howpublished = {NASA Image Article AS13-62-8929},
  year         = {2008},
  url          = {https://www.nasa.gov/image-article/apollo-13-lunar-module-mail-box/},
}

@book{pavlovic2025security,
  author    = {Dusko Pavlovic and Peter-Michael Seidel},
  title     = {Security Science ({SecSci}): Basic Concepts and Mathematical Foundations},
  year      = {2025},
  publisher = {arXiv preprint},
  note      = {arXiv:2504.16617v1 [cs.CR]},
  url       = {https://arxiv.org/abs/2504.16617},
}

@manual{modbus2012protocol,
  author       = {{Modbus Organization}},
  title        = {Modbus Application Protocol Specification V1.1b3},
  organization = {Modbus.org},
  year         = {2012},
  url          = {https://modbus.org/docs/Modbus_Application_Protocol_V1_1b3.pdf},
}

@article{lamport1974glitch,
  author  = {Leslie Lamport and Richard Palais},
  title   = {On the Glitch Phenomenon},
  year    = {1974},
  note    = {Originally rejected by IEEE Trans. on Computers (1976); republished as arXiv:1704.01154},
  url     = {https://lamport.azurewebsites.net/pubs/glitch.pdf},
}

@article{anderson1991glitch,
  author  = {James H. Anderson and Mohamed G. Gouda},
  title   = {A New Explanation of the Glitch Phenomenon},
  journal = {Acta Informatica},
  volume  = {28},
  number  = {4},
  pages   = {297--309},
  year    = {1991},
  doi     = {10.1007/BF01893884},
}

@misc{lolbins2025crowdstrike,
  author       = {{CrowdStrike}},
  title        = {2025 Global Threat Report: Living Off the Land Attacks},
  year         = {2025},
  url          = {https://www.crowdstrike.com/cybersecurity-101/cyberattacks/living-off-the-land-attack/},
}

@inproceedings{giller2015glitching,
  author    = {Brett Giller},
  title     = {Implementing Practical Electrical Glitching Attacks},
  booktitle = {Black Hat Europe},
  year      = {2015},
}

@misc{ta505lotl,
  author       = {{Proofpoint}},
  title        = {{TA505} Phishing Campaign: Living Off the Land},
  year         = {2018},
  url          = {https://www.proofpoint.com/us/threat-reference/living-off-the-land-attack},
}

\end{document}